\documentclass{amsart}[11pt]
\usepackage{amsmath, amsfonts, amsthm, amssymb}
\usepackage{subfigure}
\usepackage{verbatim}
\usepackage{hyperref}
\usepackage{color}
\usepackage{enumerate}
\usepackage[dvips]{graphicx}
\numberwithin{equation}{section}
\newtheorem{theorem}{Theorem}[section]

\newtheorem{lemma}[theorem]{Lemma}
\newtheorem{proposition}[theorem]{Proposition}

\theoremstyle{definition}
\newtheorem{definition}[theorem]{Definition}
\newtheorem{remark}[theorem]{Remark}
\newtheorem*{notation}{Notation}

\newcommand{\ind}{1\hspace{-2.1mm}{1}}

\newcommand{\D}{\mathrm{d}}
\newcommand{\RR}{\mathbb{R}}

\newcommand{\eps}{\varepsilon}
\newcommand{\Dd}{\mathcal{D}}

\newcommand{\E}{\mathrm{e}}

\begin{document}

\title{Asymptotic arbitrage in the Heston model}
\author{Fatma Haba}
\address{Department of Mathematics, Universit\'e Tunis El Manar}
\email{haba\_fatma@hotmail.fr}

\author{Antoine Jacquier}
\address{Department of Mathematics, Imperial College London}
\email{a.jacquier@imperial.ac.uk}

\date{\today}
\thanks{The authors are indebted to the anonymous referee for his/her helpful suggestions.}

\begin{abstract}
In the context of the Heston model, we establish a precise link between the set of equivalent martingale measures, 
the ergodicity of the underlying variance process and the concept of asymptotic arbitrage proposed in Kabanov-Kramkov~\cite{Kabanov}
and in F\"ollmer-Schachermayer~\cite{Fol}.
\end{abstract}

\maketitle

\section{Introduction}

The concept of arbitrage is the cornerstone of modern mathematical finance, and
several versions of the so-called fundamental theorem of asset pricing have been proved over the past two decades,
see for instance~\cite{Del} for an overview.
A version of it essentially states that absence of arbitrage is equivalent to the existence of an equivalent martingale measure
under which discounted asset prices are true martingales.
This then allows the use of `martingale models' (either continuous or with jumps) as underlying dynamics for option pricing.
In practice, should short-term arbitrages arise---due to some market discrepancies---they are immediately exploited by traders,
and market liquidity therefore acts as an equilibrium agent, to prevent them from occurring significantly.
It can be argued, however, that one may generate long-term riskless profit, when the time horizon tends to infinity.
This turns out to be the case in most models used in practice, 
and the existence and nature of such infinite horizon asymptotic arbitrage opportunities have been studied in~\cite{Du,Ire,Rok}.

Among the plethora of models used and analysed both in practice and in theory,
stochastic volatility models have proved to be very flexible and suitable for pricing and hedging.
Due to its affine structure, the Heston model~\cite{Heston} has gained great popularity
among practitioners for equity and FX derivatives modelling (see~\cite{Fou,Gatheral} for a detailed account).
Because of the correlation between the asset price and the underlying volatility, the market is incomplete,
and the Heston model admits an infinity of equivalent martingale measures.
Its affine structure allows us to study precisely the existence (or absence) of asymptotic arbitrage.
Specifically, we shall endeavour to understand how the parameters of the model influence
the nature---such as its speed and existence---of this asymptotic arbitrage.
Of particular interest will be the link between asymptotic arbitrage and the ergodicity of the underlying variance process.
In~\cite{Fol} the authors proved under suitable regularity conditions that price processes with a non-trivial market price of risk
(see Definition~\ref{def:MarketPrice}) allow for asymptotic arbitrage (with linear speed).
Using the theory of large deviations, we shall show that $S$ may allow for such arbitrage even if
it does not admit an average squared market price of risk.

The organisation of this paper is as follows: all the notations and definitions are given in Section~\ref{sec:Notations}.
Asymptotic arbitrage in the Heston model is studied in Section~\ref{sec:Main};
the main contribution of this paper is Theorem~\ref{thm:AsymptArb},
which identifies sufficient (and sometimes necessary) conditions on the set of equivalent martingale measures under which
asymptotic arbitrage occur with linear speed.
These conditions are different from those in Proposition~\ref{Propslowerspeed} in which we study
the role of the ergodicity of the variance process on the existence of asymptotic arbitrage with slower speed.

\section{Notations and definitions}\label{sec:Notations}

Let $(\Omega,\mathcal{F},\mathbb{F},\mathbb{P})$ be a filtered probability space where the filtration $\mathbb{F}=(\mathcal{F}_t)_{t\geq0}$ satisfies the usual conditions,
$S=\E^{X}$ model a risky security under an equivalent martingale measure,
and let $\mathcal{H}$ denote the class of predictable, $S$-integrable admissible processes.
We define for each $t\geq 0$ the sets
of strategies~$K_t$
and of equivalent local martingales $\mathcal{M}_t^{e}(S)$ by
$$K_t:=\left\{\int_0^t H_s \D S_s: H\in \mathcal{H}\right\}
\quad
\text{and}
\quad
\mathcal{M}_t^{e}(S):=
\left\{\mathbb{Q}\sim\mathbb{P}: \text{ such that }(S_u)_{0\leq u\leq t}\text{ is a local }
\mathbb{Q}\text{-martingale}\right\}.
$$
We shall always assume that $\mathcal{M}_t^{e}(S)$ is not empty.
Furthermore, for any set $A$ in $\Omega$,
we shall denote by $A^c:=\Omega\setminus A$ its complement.

\subsection{Asymptotic arbitrage}
We are interested here in specific forms of arbitrage (asymptotic arbitrage),
first introduced by Kabanov and Kramkov~\cite{Kabanov}, and refined recently
by F\"{o}llmer and Schachermayer~\cite{Fol}.
The following definition of a $(\eps_1,\eps_2)$ arbitrage is taken from the latter~\cite[Proposition 2.1]{Fol}:
\begin{definition}\label{def:eps1eps2Arb}
The process $S$ admits an $(\eps_1,\eps_2)$-arbitrage 
if for $(\eps_1,\eps_2)\in (0,1)^2$, there exists $X_t \in K_t$ such that
\begin{enumerate}[(i)]
\item $X_t \geq-\eps_2$ $\mathbb{P}$-almost surely;
\item $\mathbb{P}(X_t \geq 1-\eps_2)\geq1-\eps_1$.
\end{enumerate}
\end{definition}
This means that the maximal loss of the trading strategy, yielding the wealth $X_t$ at time $t$,
is bounded by $\eps_2$ and, with probability $1-\eps_1$,
the terminal wealth $X_t$ equals at least $1-\eps_2$.
Let us consider the following  slightly weaker version,
which imposes less stringent restrictions on the maximal loss:
\begin{definition}\label{def:eps1eps2ArbNew}
Let $(e_1,e_2)\in (0,1)^2$.
We say that the process $S$ admits a partial $(e_1,e_2)$-arbitrage (up to time $t>0$)
if for any $(\eps_1,\eps_2)\in (e_1,1)\times (e_2,1)$, there exists $X_t \in K_t$ such that
\begin{enumerate}[(i)]
\item $X_t \geq-\eps_2$ $\mathbb{P}$-almost surely;
\item $\mathbb{P}(X_t\geq 1-\eps_2)\geq1-\eps_1$.
\end{enumerate}
\end{definition}
Obviously, Definition~\ref{def:eps1eps2ArbNew} is equivalent to Definition~\ref{def:eps1eps2Arb} when $e_1=e_2=0$.
However, as we shall see later, (partial) asymptotic arbitrage may only appear for $e_1$ exponentially small, 
but necessarily equal to zero.
When $e_1$ is not null, this alternative definition may also be seen as a mid-point characterisation between  Definition~\ref{def:eps1eps2Arb}
and  Definition~\ref{def:ExpArb} below.
The latter in particular characterises the notion of asymptotic exponential arbitrage with exponentially decaying failure probability,
first proposed in~\cite{Fol} and studied later in~\cite{Bidima} and~\cite{Du}.

\begin{definition}\label{def:ExpArb}
The process $S$ allows for asymptotic exponential arbitrage with exponentially decaying failure probability
if there exist $t_0\in(0,\infty)$ and constants $C,\lambda_1,\lambda_2>0$
such that for all $t\geq t_0$, there is $X_t\in K_t$ satisfying
\begin{enumerate}[(i)]
\item $X_t \geq -\E^{-\lambda_2 t}$ $\mathbb{P}$-almost surely;
\item $\mathbb{P}(X_t \leq \E^{\lambda_2 t}) \leq C\E^{-\lambda_1 t}$.
\end{enumerate}
\end{definition}
\begin{remark}
We see that there is a relation between Definition~\ref{def:ExpArb} and Definition~\ref{def:eps1eps2ArbNew}. 
If the process~$S$ allows for asymptotic exponential arbitrage with exponentially decaying failure probability, 
then it admits a partial $(e_1,e_2)$-arbitrage with $e_1=C\E^{-\lambda_1 t}$ and $e_2=\E^{-\lambda_2 t}$.
Indeed, for any $\eps_2\in(\E^{-\lambda_2 t},1)$ we have $X_t\geq -\E^{-\lambda_2 t}\geq -\eps_2$.
Now for large $t$, we have $1-C\E^{-\lambda_1t}\leq\mathbb{P}(X_t>\E^{\lambda_2 t})\leq\mathbb{P}(X_t\geq 1-\E^{-\lambda_2 t})=\mathbb{P}(X_t\geq 1-\eps_2)$. Then for any $\eps_1\in(C\E^{-\lambda_1t},1)$ with $C\in(0,1)$, we have $\mathbb{P}(X_t\geq 1-\eps_2)\geq 1-\eps_1$.
\end{remark}
\begin{definition}\label{def:MarketPrice}
Let $f:\RR_+^*\to\RR_+^*$ be a smooth function such that $\lim_{t\uparrow \infty}f(t)=+\infty$.
The process~$S$ is said to have an average squared market price of risk $\gamma$
above the threshold $c>0$ with speed~$f(t)$ if
$\mathbb{P}\left(f(t)^{-1}\int_{0}^{t}\gamma^2(s) \D s<c\right)$ tends to zero
as $t$ tends to infinity.
\end{definition}

\subsection{Stochastic volatility models}
We consider here the Heston stochastic volatility model, namely the unique strong solution
to the stochastic differential equations~\eqref{eq:SDEHeston} below.
As is well-known~\cite{Ber}, there may not be a unique risk-neutral martingale measure for this .
The following SDEs are therefore understood under one such risk-neutral measure $\mathbb{Q}$.
\begin{equation}\label{eq:SDEHeston}
\begin{array}{rll}
\D S_t / S_t & = \mu \D t+ \sqrt{V_t}\left(\rho \D W_1(t)+\sqrt{1-\rho^{2}}\D W_2(t)\right),\quad & S_0=1\\
\D V_t & = (a-bV_t)\D t+\sqrt{2\sigma V_t}\D W_1(t),\quad & V_0>0,
\end{array}
\end{equation}
where $W_1$ and $W_2$ are independent $\mathbb{Q}$-Brownian motions, $a,\sigma>0$,
$\mu, b\in\mathbb{R}$ and $|\rho|<1$.
The class of equivalent martingale measures $\mathbb{Q}$ can be considered in terms of the Radon-Nikodym derivatives
\begin{equation}\label{eq:Z}
Z_t
 := \left.\frac{\D\mathbb{Q}}{\D\mathbb{P}}\right|_{\mathcal{F}_t}
 = \exp\left\{-\left(\int_0^t\gamma_1(s)\D W_1(s)+\int_0^t\gamma_2(s) \D W_2(s)\right)
-\frac{1}{2}\left(\int_0^t\gamma_1^{2}(s)\D s+\int_0^t\gamma_2^{2}(s)\D s\right)\right\}.
\end{equation}
The condition $\mu-r=\sqrt{V_t}(\rho\gamma_1(t)+\sqrt{1-\rho^2}\gamma_2(t))$ is necessary for an equivalent local martingale measure to exist, and ensures that the discounted stock price is a local martingale;
here $r$ denotes the constant risk-free rate.
Since $Z$ is a positive local martingale with $Z_0=1$, it is a supermartingale, and a true martingale if and only if  $\mathbb{E}(Z_t)=1$.
For the Heston stochastic volatility model we obtain, for any~$\lambda\in\RR$,
\begin{equation}\label{eq:DefGamma}
\gamma_1(t) = \lambda\sqrt{V_t}
\qquad\text{and}\qquad
\gamma_2(t) = \frac{1}{\sqrt{1-\rho^2}}\left(\frac{\mu-r}{\sqrt{V_t}}-\lambda \rho\sqrt{V_t}\right).
\end{equation}

\section{Main results}\label{sec:Main}
For any $(\alpha, \beta,\delta)\in\RR^3$, we introduce the process $(X_t^{\alpha,\beta,\delta})_{t\geq 0}$ defined (pathwise) by
\begin{equation}\label{eq:XProcess}
X_t^{\alpha,\beta,\delta}:=\alpha V_t+\beta\int_0^t V_s\D s+\delta\int_0^t  V_s^{-1}\D s,
\quad\text{for any }t\geq 0,
\end{equation}
where $V$ is the Feller diffusion for the variance in~\eqref{eq:SDEHeston}.
We shall always assume that $\beta$ and $\delta$ are not both null simultaneously.
In that case, $X$ is simply the Feller diffusion,
and its density is known in closed form~\cite[Part 1, Chapter 6.3]{JYCBook}.
The large-time behaviour of $X$ will play a key role in determining average squared market prices of risk,
and the case $\beta=\delta=0$ will never occur, so this assumption does not entail any loss of generality here.
Define the real interval $\Dd_{\beta,\delta}$ by
\begin{equation}\label{eq:DomainD}
\Dd_{\beta,\delta}=
\left\{
\begin{array}{ll}
\displaystyle \left[\frac{(a-\sigma)^2 }{4\sigma\delta}, \frac{b^2 }{4\sigma\beta}\right], & \text{if } \beta>0,\;\delta<0,\\
\displaystyle \left(-\infty,\frac{(a-\sigma)^2 }{4\sigma\delta}\wedge\frac{b^2 }{4\sigma\beta}\right], & \text{if } \beta>0,\; \delta>0,\\
\displaystyle \left[\frac{b^2 }{4\sigma\beta}, \frac{(a-\sigma)^2 }{4\sigma\delta}\right], & \text{if } \beta<0,\; \delta>0,\\
\displaystyle \left[\frac{(a-\sigma)^2}{4\sigma\delta}\vee\frac{b^2}{4\sigma\beta},+\infty\right), & \text{if } \beta<0, \;\delta<0.
\end{array}
\right.
\end{equation}
Whenever $\beta\delta=0$, we define $\Dd_{\beta,\delta}$ by taking the limits of the interval (a closed bound becoming open if it becomes infinite),
where we use the slight abuse of notation $"1/0=\infty"$, i.e.
$\Dd_{\beta,\delta}=\left(-\infty,\frac{b^2}{4\sigma\beta}\right]$ if $\beta>0$ and $\delta=0$,
$\Dd_{\beta,\delta}=\left[\frac{b^2}{4\sigma\beta} ,+\infty\right)$ if $\beta<0$ and $\delta=0$.
Let us further define the function $\Lambda^{\beta,\delta}:\Dd_{\beta,\delta}\to\RR$ by
\begin{equation}\label{eq:Lambda}
\Lambda^{\beta,\delta} (u) =
\left\{
\begin{array}{ll}
\displaystyle \frac{ba}{2\sigma}-\frac{1}{2\sigma}\sqrt{((a-\sigma)^2-4\sigma\delta u)(b^2-4\sigma\beta u)}-\frac{1}{2}\sqrt{b^2-4\sigma\beta u},
 & \text{if }\delta\ne 0,\\
\\
\displaystyle \frac{a}{2\sigma}\left(b-\sqrt{b^2-4\sigma\beta u}\right),
 & \text{if }\delta=0.
\end{array}
\right.
\end{equation}
In the case $\delta \ne 0$ above, we further impose the condition $a>\sigma$
for the definition of the function~$\Lambda^{\beta,\delta}$.
\begin{remark}
It may be surprising at first that the function $\Lambda^{\beta,\delta}$ related---in some sense defined precisely below---does
not depend on $\alpha$.
This function actually describes the large-time behaviour of the process~$X^{\alpha,\beta,\delta}$.
Since the variance process $V$ is strictly positive almost surely (by the Feller condition imposed above),
the term $\int_0^t V_s \D s$ clearly dominates $V_t$ for any $t$, which explains why $\alpha$
bears no influence on~$\Lambda^{\beta,\delta}$.
The condition $a>\sigma$ imposed above in the case $\delta\ne 0$ should not surprise the reader since
this is nothing else than the Feller condition, ensuring that the variance process never touches the origin almost surely.
\end{remark}

We further define the Fenchel-Legendre transform $\Lambda_{\beta,\delta}^*:\RR\to\RR_+$ of $\Lambda^{\beta,\delta}$ by
\begin{equation}\label{eq:LambdaStar}
\Lambda_{\beta,\delta}^*(x):=\sup_{u\in\Dd_{\beta,\delta}}\{ux-\Lambda^{\beta,\delta}(u)\}.
\end{equation}

\begin{notation}
Whenever $\beta=0$ or $\delta=0$, we shall drop the subscript and write respectively $\Lambda^\delta$ or $\Lambda^\beta$.
The same rule shall apply for the Fenchel-Legendre transforms and their respective domains.
\end{notation}

In general, $\Lambda^*_{\beta,\delta}$ does not have a closed-form representation.
However when $\delta$ is null---which shall be of interest for us---it actually does,
and a straightforward computation shows that
\begin{equation}\label{eq:Lambda*Delta0}
\Lambda^*_{\beta,0}(x) \equiv \Lambda^*_{\beta}(x) = \frac{(bx-a\beta )^2}{4\sigma |\beta x|},
\quad\text{for all }x\in\RR^*.
\end{equation}
In that case, the function $\Lambda^*_{\beta}$ is strictly convex on $\RR_+^*$ (respectively on $\RR_-^*$) with a unique minimum
attained at $|a\beta/b|$ (resp. at $-|a\beta/b|$).
In particular, if $b\beta>0$, then $\Lambda^*_{\beta}(|a\beta/b|) = 0$ and  $\Lambda^*_{\beta}(x)>0$
for all $x\in\RR_+^*\setminus\{|a\beta/b|\}$.
Symmetric statements hold on $\RR_-$.

\subsection{The large deviations case}\label{sec:Limitl}
In this section, we prove asymptotic arbitrage results (with linear speed) for the stock price process;
we shall in particular observe that the ergodicity of the variance process plays a key role.
We first start with the following lemma (proved in Appendix~\ref{App:Appendix}), which will be used heavily in the remaining of the paper.
For precise definitions of large deviations principles (LDP), we refer the reader to the excellent monograph~\cite{DZ};
we shall use the non-standard terminology `partial large deviations principles' if an LDP holds only on subsets of the real line.


\begin{lemma}\label{lem:triplet}
As $t$ tends to infinity, the family $(t^{-1}X_t^{\alpha,\beta,\delta})_{t\geq0}$ satisfies
\begin{enumerate}[(i)]
\item a full LDP (on $\RR$) if $\beta\delta<0$;
\item a partial LDP on $\left(2\sqrt{\delta\beta},+\infty\right)$ if $\beta\geq 0$ and $\delta\geq 0$;
\item a partial LDP on $\left(-\infty,-2\sqrt{\delta\beta}\right)$ if $\beta\leq 0$ and $\delta\leq 0$;
\end{enumerate}
In each case, the rate function is $\Lambda^*_{\beta,\delta}$ and the (partial) LDP holds with speed $t^{-1}$.
\end{lemma}

In~\cite[Theorem 1.4]{Fol}, F\"ollmer and Schachermayer proved that if the stock price process has an average market price of risk above a threshold then asymptotic arbitrage holds.
Using the large deviations principle proved above, we first show that $S$ does not always admit an average market price of risk for $\gamma_1$ (Proposition~\ref{prop:AvgSqGamma1}) 
or $\gamma_2$ (Proposition~\ref{prop:AvgSqGamma2}) above any threshold.
This is in particular so when the variance process is not ergodic ($b\leq 0$).
This however---as proved in Theorem~\ref{thm:AsymptArb} below---does not preclude absence of asymptotic arbitrage.

\begin{proposition}\label{prop:AvgSqGamma1}
Fix $\lambda\geq 0$ and $c>0$.
The stock price process does not satisfy an average squared market price of risk $\gamma_1$ above the threshold $c$ with speed $t$ if either (i) $b\leq 0$ or (ii) $b>0$ and $c > a\lambda^2/b$.
\end{proposition}
\begin{proof}
Note first that $\lambda=0$ implies $\gamma_1\equiv 0$ and hence
$\mathbb{P}(t^{-1}\int_0^t \gamma_1^2(s)\D s< c)=1$ for all $t>0$,
so that the proposition is trivial.
Assume from now on that $\lambda\ne 0$ and
let $c$ be an arbitrary strictly positive real number.
The definition of $\gamma_1$ in~\eqref{eq:DefGamma} implies
$\mathbb{P}(t^{-1}\int_0^t \gamma_1^2(s)\D s\geq c)
 = \mathbb{P}(t^{-1}\int_0^t V_s\D s\geq c/\lambda^2)
=\mathbb{P}(t^{-1}X_t^{0,1,0} \geq c/\lambda^2)$.
From Lemma~\ref{lem:triplet}, the family $(t^{-1}X_t^{0,1,0})_{t\geq0}$
satisfies a LDP on $\mathbb{R}_+^*$ with rate function~$\Lambda^*_{1,0}$.
Hence
\begin{equation*}
\limsup_{t\uparrow+\infty}\frac{1}{t}\log\mathbb{P}\left(X_t^{0,1,0}
\geq \frac{c}{\lambda^2}\right)\leq -\inf_{\left\{x\geq c/\lambda^2\right\}}\Lambda^*_{1,0}(x)
\end{equation*}
When $b\leq 0$, $\inf_{\left\{x\geq c/\lambda^2\right\}}\Lambda^*_{1,0}(x)$ is strictly positive for all $c>0$.
Thus
$\mathbb{P}(t^{-1}X_t^{0,1}\geq c/\lambda^2)$ converges to zero as $t$ tends to infinity,
which in turn implies that
$\mathbb{P}(t^{-1}\int_0^t \gamma_1^2(s)\D s< c)$ converges to $1$ as $t$ tends to infinity,
and statement~(i) in the proposition follows.
When $b>0$, consider the case $c > a\lambda^2/b$;
then $\inf_{\left\{x\geq c/\lambda^2\right\}}\Lambda^*_{1,0}(x)$ is strictly positive and we end up with the same as in the case $b\leq0$ which proves statement~(ii) in the proposition.
\end{proof}

\begin{proposition}\label{prop:AvgSqGamma2}
Fix $\lambda\geq 0$ and let $c>0$.
The stock price process does not satisfy an average squared market price of risk $\gamma_2$ above the threshold $c$ with speed $t$
if any of the following conditions hold:
\begin{enumerate}[(i)]
\item $\lambda\rho(\mu-r)> 0$;
\item $\lambda\rho(\mu-r)<0$ and $c>-4\lambda\rho(\mu-r)/(1-\rho^2)$;
\item $\lambda\rho\ne 0$, $\mu=r$ and $b\leq 0$;
\item $\lambda\rho\ne 0$, $\mu=r$, $b>0$ and $c>a\lambda^4\rho^2/(b(1-\rho^2))$;
\item $\lambda\rho = 0$;
\end{enumerate}
\end{proposition}
\begin{remark}
Note that the case $\lambda\rho=0$ precisely corresponds to the case of a complete market.
\end{remark}
\begin{proof}
Let $c$ be an arbitrary strictly positive real number.
Note first that if $\lambda\rho=0$, and $\mu=r$, then $\gamma_2\equiv 0$ and hence
$\mathbb{P}\left(t^{-1}\int_0^t \gamma_2^2(s)\D s< c\right)=1$ for all $t>0$.
If $\mu\ne r$, then
$$
\mathbb{P}\left(\frac{1}{t}\int_0^t\gamma_2^2(s)\D s\geq c\right)
=\mathbb{P}\left(\frac{1}{t}\int_0^t\frac{\D s}{V_s}\geq\frac{1-\rho^2}{(\mu-r)^2}c\right)
=\mathbb{P}\left(\frac{X_t^{0,0,1}}{t}\geq\frac{1-\rho^2}{(\mu-r)^2}c\right),$$
and Lemma~\ref{lem:ratefunction} implies that $\Lambda^*_{0,1}$ is strictly positive, so that~(v) follows.
Assume  now  that $\lambda\rho\ne 0$ and $\mu\ne r$.
The definition of $\gamma_2$ in~\eqref{eq:DefGamma} implies that
\begin{align*}
\mathbb{P}\left(\frac{1}{t}\int_0^t\gamma_2^2(s)\D s\geq c\right)
&  = \mathbb{P}\left(\frac{(\mu-r)^2}{1-\rho^2}\frac{1}{t}\int_0^t\frac{\D s}{V_s}
+\frac{\lambda^2\rho^2}{1-\rho^2}\frac{1}{t}\int_0^t V_s \D s
\geq c+\frac{2\rho\lambda(\mu-r)}{1-\rho^2}\right)\\
& = \mathbb{P}\left(\frac{X_t^{0,\beta,\delta}}{t}\geq c+\frac{2\rho\lambda(\mu-r)}{1-\rho^2}\right),
\end{align*}
where $\beta=\frac{(\mu-r)^2}{1-\rho^2}>0$,
$\delta=\frac{\lambda^2\rho^2}{1-\rho^2}>0$,
and where~$X^{0,\beta,\delta}$ is defined in~\eqref{eq:XProcess}.
By Lemma~\ref{lem:triplet}, the family
$(X_t^{0,\beta,\delta}/t)_{t > 0}$
satisfies a large deviations principle on $(2\sqrt{\delta\beta}, +\infty)$ with rate function $\Lambda^*_{\beta,\delta}$, i.e.
$$
\limsup_{t\uparrow+\infty}t^{-1}\log\mathbb{P}\left(\frac{1}{t}\int_0^t\gamma_2^2(s)\D s\geq c\right)
\leq -\inf\left\{\Lambda^*_{\beta,\delta}(x): x\geq c+\frac{2\rho\lambda(\mu-r)}{1-\rho^2}\right\}.
$$
When $\lambda\rho(\mu-r) > 0$, $\left[c+\frac{2\rho\lambda(\mu-r)}{1-\rho^2},+\infty\right)$ is a subset of
$\left(2\sqrt{\beta\delta},+\infty\right)$, and (i) follows immediately from Lemma~\ref{lem:ratefunction}.
When $\lambda\rho(\mu-r)<0$, the interval
$\left[c+\frac{2\rho\lambda(\mu-r)}{1-\rho^2},+\infty\right)$ is a subset of $\left(2\sqrt{\beta\delta},+\infty\right)$
if and only if $c>-\frac{4\rho\lambda(\mu-r)}{1-\rho^2}>0$.
Since $\beta\delta=\frac{\lambda^2\rho^2(\mu-r)^2}{(1-\rho^2)}>0$,
Lemma~\ref{lem:ratefunction} implies that $\Lambda^*_{\beta,\delta}(x)>0$
for any $x>2\sqrt{\beta\delta}=\frac{2|\lambda\rho(\mu-r)|}{1-\rho^2}$.
Therefore, $\mathbb{P}\left(X_t^{0,\beta,\delta}/t\geq c+\frac{2\rho\lambda(\mu-r)}{1-\rho^2}\right)$ converges to zero as $t$ tends to infinity.
Then $\mathbb{P}(t^{-1}\int_0^t \gamma_2^2(s)\D s< c)$ converges to one as $t$ tends to infinity.
Assume that $\lambda\rho\ne 0$ and $\mu=r$.
The definition of $\gamma_2$ in~\eqref{eq:DefGamma} implies that
$$
\mathbb{P}\left(\frac{1}{t}\int_0^t\gamma_2^2(s)\D s\geq c\right)
=\mathbb{P}\left(\frac{1}{t}\int_0^t V_s\D s\geq \frac{1-\rho^2}{\lambda^2\rho^2}c\right)
=\mathbb{P}\left(\frac{X_t^{0,1,0}}{t}\geq \frac{1-\rho^2}{\lambda^2\rho^2}c\right),$$
and (iii) and (iv) then follow from Proposition~\ref{prop:AvgSqGamma1}.
\end{proof}

We can now move on to our main theorem, which proves a partial arbitrage for the stock price process~$S$.
\begin{theorem}\label{thm:AsymptArb}
There exists $\gamma^*>0$ such that for all 
$\lambda<-\frac{b}{\sqrt{2\sigma}}-\frac{\gamma^*\sqrt{2\sigma}}{2a}$, 
$S_t$ admits a partial $(e_1,1/2)$-arbitrage
for~$t$ large enough, 
with 
$e_1 = \exp\left[\frac{\lambda V_0}{\sqrt{2\sigma}}+\left(\frac{a\lambda}{\sqrt{2\sigma}}+\gamma^* + \Lambda^{\beta}(1)\right)t\right]$.
\end{theorem}

\begin{remark}\label{rem:Lambda0}
The threshold $e_1$ has the form $e_1\sim \E^{-\lambda_1 t}$ for some $\lambda_1>0$, 
which links partial arbitrage to exponentially decaying failure probability characterised in Definition~\ref{def:ExpArb}.
Note that we could slightly relax the constraint on $\lambda$, making the latter time-dependent, 
because we only need to ensure that $e_1\in (0,1)$.
Since we are only interested in large~$t$, this is however not essential here.
The sufficient condition on $\lambda$ is not necessary:
for $\lambda=0$ and $\mu=r$, $Z_t = 1$ almost surely for all $t\geq 0$, and
$\mathbb{P}\left(Z_t\geq \E^{-\gamma t}\right)=1$ for any $\gamma>0$.

\end{remark}

\begin{proof}
Let $\gamma>0$ and define the set $A_{\lambda,t}:=\{Z_t\geq \E^{-\gamma t}\}\in\mathcal{F}_t$.
Since the processes $W_2$ and $V$ are independent, the tower property for conditional expectation implies
$\mathbb{E}(Z_t)=\mathbb{E}\left(\E^{-\int_0^t\gamma_1(s)\D W_1(s)-\frac{1}{2}\int_0^t\gamma_1^2(s)\D s}\right)$.
Markov's inequality therefore yields
\begin{align*}
\mathbb{P}(A_{\lambda,t})
& \leq\frac{\mathbb{E}(Z_t)}{\exp(-\gamma t)}
 = \frac{\mathbb{E}\left[\exp\left(-\int_0^t\gamma_1(s)\D W_1(s)-\frac{1}{2}\int_0^t\gamma_1^2(s)\D s\right)\right]}{\E^{-\gamma t}}\\
& =\exp\left(\frac{\lambda V_0}{\sqrt{2\sigma}}+\frac{a\lambda t}{\sqrt{2\sigma}}
+\gamma t\right)
\mathbb{E}\left[\exp\left(-\frac{\lambda V_t}{\sqrt{2\sigma}}-\left(\frac{b\lambda}{\sqrt{2\sigma}}+\frac{\lambda^2}{2}\right)\int_0^tV_s\D s\right)\right]\\
& =\exp\left[\frac{\lambda V_0}{\sqrt{2\sigma}}+\left(\frac{a\lambda}{\sqrt{2\sigma}}+\gamma \right)t\right]
\Lambda_t^{\alpha,\beta}(t),
\end{align*}
where $\alpha=-\frac{\lambda}{\sqrt{2\sigma}}$ and $\beta=-\frac{b\lambda}{\sqrt{2\sigma}}-\frac{\lambda^2}{2}$.
From the proof of Lemma~\ref{lem:triplet}, we know that
$t^{-1}\log\Lambda^{\alpha,\beta}_t(t)$ converges to $ \Lambda^{\beta}(1)$,
which implies that for any $\eta>0$ there exists $\tilde{t}>0$ such that for any $t>\tilde{t}$,
$
\E^{\left(\Lambda^{\beta}(1)-\eta\right)t}\leq \Lambda_t^{\alpha,\beta}(t)\leq \E^{\left(\Lambda^{\beta}(1)+\eta\right)t}.
$
Therefore, for any $t>\tilde{t}$,
$$
 \mathbb{P}(Z_t\geq \E^{-\gamma t}) \leq
\exp\left[\frac{\lambda V_0}{\sqrt{2\sigma}}+\left(\frac{a\lambda}{\sqrt{2\sigma}}+\gamma + \Lambda^{\alpha,\beta}(1)+\eta\right)t\right].
$$
Since $\eta$ can be chosen as small as desired, we simply need to prove that
$\frac{a\lambda}{\sqrt{2\sigma}}+\gamma + \Lambda^{\beta}(1)<0$.
From Appendix~\ref{App:Appendix}, this inequality is satisfied whenever 
$\lambda<-\frac{b}{\sqrt{2\sigma}}-\frac{\gamma\sqrt{2\sigma}}{2a}$.
Let $\eps_1 :=\exp\left[\frac{\lambda V_0}{\sqrt{2\sigma}}+\left(\frac{a\lambda}{\sqrt{2\sigma}}+\gamma + \Lambda^{\beta}(1)\right)t\right]$ and $\eps_2\in(0,1)$.
The random variable
$Y_t := \eps_2\ind_{A_{\lambda,t}^c} - \eps_2\frac{\mathbb{Q}(A_{\lambda,t}^c)}{\mathbb{Q}(A_{\lambda,t})}\ind_{A_{\lambda,t}}$ satisfies
$\mathbb{E}_{\mathbb{Q}}(Y_t)=0$ for all $\mathbb{Q}\in\mathcal{M}_t^e(S)$,
and hence $Y_t\in K_t$ by~\cite[Proposition 2.2.10]{Del}.
On $A_{\lambda,t}^c$, clearly $Y_t=\eps_2\geq-\eps_2$.
Now, the family $(\{Z_t\geq \E^{-\gamma t}\})_{\gamma>0}$ forms an increasing sequence of sets,
and hence $\mathbb{Q}(Z_t\geq\E^{-\gamma t})$ is an increasing function of $\gamma$.
Therefore there exists $\gamma^*>0$ such that for any $\gamma>\gamma^*$,
$\mathbb{Q}(A_{\lambda,t}^c)\leq \mathbb{Q}(A_{\lambda,t})$, and therefore $Y_t\geq -\eps_2$ on $A_{\lambda,t}$.
This then yields $Y_t\geq -\eps_2$ almost surely.
Finally, since $\mathbb{P}(Y_t=\eps_2)=\mathbb{P}(A_{\lambda,t}^c)\geq 1-\eps_1$,
then $\mathbb{P}(Y_t\geq 1-\eps_2)\geq\mathbb{P}(Y_t=\eps_2)\geq 1-\eps_1$ for $\eps_2\in[1/2,1)$.
Also, $\mathbb{P}(Y_t<1-\eps_2)=\mathbb{P}(A_{\lambda,t})\leq\eps_1$,
and hence $S$ allows for partial $(e_1,1/2)$-arbitrage in the sense of Definition~\ref{def:eps1eps2ArbNew}.
\end{proof}
\begin{remark}
$(\eps_1,\eps_2)$-arbitrage in the sense of the Definition~\ref{def:eps1eps2Arb} is harder to prove here since it would be equivalent to the existence of a set~$A_t\in\mathcal{F}_t$, with $\mathbb{P}(A_t)\leq \eps_1$
such that
$\mathbb{Q}(A_t)\geq 1-\eps_2$ holds for all $\mathbb{Q}\in\mathcal{M}^e(S)$
(by~\cite[Proposition 2.1]{Fol}).
Since the random variable $Z_t$ above depends on the parameter~$\lambda$,
then the proof above shows that both $\mathbb{P}(A_{\lambda,t})\leq\eps_1$
and $\mathbb{Q}(A_{\lambda,t})\geq 1-\eps_2$ hold
only when $\lambda<-\frac{b}{\sqrt{2\sigma}}-\frac{\gamma\sqrt{2\sigma}}{2a}$;
but the measure $\mathbb{Q}$ also depends on $\lambda$, and hence $(\eps_1,\eps_2)$-arbitrage may not hold.
\end{remark}


 \subsection{Case $t/f(t)$ and $f(t)$ tend to infinity as $t$ tends to infinity}
 Let $b>0$, in which case the variance process is ergodic and  its stationary distribution~$\pi$
is a Gamma law with shape parameter $a/\sigma$ and scale parameter $\sigma/b$;
namely $t^{-1}\int_0^t h(V_s) \D s$ converges to $\int_{\mathbb{R}}h(x)\pi(\D x)$ almost surely
for any $h\in L^1(\pi)$ (see~\cite{Yur}).
In this section, we consider a continuous function~$f:\RR^*_{+}\rightarrow \RR_{+}$ such that
$t/f(t)$ tends to infinity as $t$ tends to infinity.
We shall prove below that (under some conditions on the risk parameter $\lambda$)
the ergodicity of the variance ensures that~$S$ allows an asymptotic arbitrage with sublinear speed~$f(t)$.

\begin{proposition}\label{prop:limProbErgodic}
The stock price process $S$ in~\eqref{eq:SDEHeston} has an average squared market price of risk $\gamma_1$
above the threshold $a\lambda^2/b$ with speed $f(t)$.
If furthermore $a>\sigma$ and $\lambda\rho(\mu-r)\leq 0$, then there exists $c_2>0$ such that $S$
has an average squared market price of risk $\gamma_2$ above the threshold $c_2$ with speed $f(t)$.
\end{proposition}
\begin{remark}\label{rem:C2}
As the proof shows, we can actually be more precise regarding the threshold~$c_2$:
\begin{itemize}
\item if $\mu=r$, then $c_2 = \frac{a\lambda^2\rho^2}{b(1-\rho^2)}$;
\item if $\mu\ne r$ and $\rho\lambda<0$, then no further condition on $c_2$ is needed;
\item if $\mu\ne r$ and $\rho\lambda=0$, then $c_2=\frac{(\mu-r)^2b}{(a-\sigma)(1-\rho^2)}$;
\end{itemize}
\end{remark}

It is rather interesting to compare this result with those of Proposition~\ref{prop:AvgSqGamma1} and Proposition~\ref{prop:AvgSqGamma2}.
Indeed, when $b>0$,
if $f(t)\equiv t$ then the stock price process does not satisfy an average squared market price
of risk $\gamma_1$ above the threshold $a\lambda^2/b$.
However, when $t/f(t)$ tends to infinity, then $S$ has an average squared market price of risk  $\gamma_1$ above the threshold $a\lambda^2/b$.
When $b>0$, $\lambda\rho\ne0$ and $\mu=r$,
if $f(t)\equiv t$ then the stock price process does not satisfy an average squared market price
of risk $\gamma_2$ above the threshold $\frac{a\lambda^4\rho^2}{b(1-\rho^2)}$,
but does so above the threshold $\frac{a\lambda^2\rho^2}{b(1-\rho^2)}$
when $t/f(t)$ tends to infinity.
Finally, when $b>0$, $\lambda\rho=0$ and $\mu\ne r$
the stock price process never satisfies an average squared market price
of risk $\gamma_2$ with speed $f(t)\equiv t$, but does
above the threshold $\frac{b(\mu-r)^2}{(1-\rho^2)(a-\sigma)}$ whenever $t/f(t)$ tends to infinity.

\begin{proof}[Proof of Proposition~\ref{prop:limProbErgodic}]
Let $f$ be as stated in the proposition.
For $b>0$, the variance process is ergodic and its stationary distribution
is a Gamma law with shape parameter $a/\sigma$ and scale parameter $\sigma/b$ (see~\cite{Yur}).
In particular,
$t^{-1}\int_0^t V_s \D s$ converges in probability to $a/b$ as $t$ tends to infinity,
and hence for any $c_1\in(0,a \lambda^2/b)$,
\begin{equation}\label{eq:Gamma1bPos}
\lim_{t\uparrow+\infty}\mathbb{P}\left(\frac{1}{t}\int_0^t\gamma_1^2(s) \D s<c_1\right)=0,
\qquad\text{and hence}\qquad
\lim_{t\uparrow+\infty}\mathbb{P}\left(\frac{1}{f(t)}\int_0^t\gamma_1^2(s) \D s<c_1\right)=0,
\end{equation}
which proves the first part of the proposition.

Consider now $\gamma_2$.
When $\mu=r$, the definitions~\eqref{eq:DefGamma} implies that
$\gamma_2=-\rho\gamma_1/\sqrt{1-\rho^2}$, and hence
$$
\lim_{t\uparrow+\infty}\mathbb{P}\left(\frac{1}{f(t)}\int_0^t\gamma_2^2(s) \D s<c_2\right)
 = \lim_{t\uparrow+\infty}\mathbb{P}\left(\frac{1}{f(t)}\int_0^t\gamma_1^2(s) \D s<\frac{(1-\rho^2)c_2}{\rho^2}\right)
$$
is equal to zero if and only if $\left(1-\rho^2\right)c_2/\rho^2\in (0,a\lambda^2/b)$, and the proposition follows.

We now assume that $\mu\ne r$.
If $a>\sigma$ we further know that (see proposition 4 in~\cite{Ala})
$t^{-1}\int_0^tV_s^{-1}\D s$ converges in probability to $b/(a-\sigma)$ as $t$ tends to infinity.
Therefore for any $c\in(0,b/(a-\sigma))$ we have
\begin{equation}\label{eq:Gamma2bPos}
\lim_{t\uparrow\infty}\mathbb{P}\left(\frac{1}{t}\int_0^t \frac{\D s}{V_s}<c\right)=0,
\qquad\text{and hence}\qquad
\lim_{t\uparrow+\infty}\mathbb{P}\left(\frac{1}{f(t)}\int_0^t \frac{\D s}{V_s}<c\right)=0.
\end{equation}
Let $c_2, c'_1, c'_2$ be three strictly positive numbers such that $c_2=c'_1+c'_2$.
The definition of $\gamma_2$ in~\eqref{eq:DefGamma} implies
\begin{align*}
\mathbb{P}\left(\frac{1}{f(t)}\int_0^t\gamma_2^2(s)\D s<c_2\right)
& = \mathbb{P}\left(\frac{1}{f(t)}\frac{(\mu-r)^2}{1-\rho^2}\int_0^t\frac{\D s}{V_s}
-\frac{2\rho\lambda(\mu-r)}{1-\rho^2}\frac{t}{f(t)}
+\frac{1}{f(t)}\frac{\lambda^2\rho^2}{1-\rho^2}\int_0^t V_s \D s<c_2\right)\\
& \leq
\mathbb{P}\left(\frac{1}{f(t)}\frac{\lambda^2\rho^2}{1-\rho^2}\int_0^t V_s\D s<c'_1\right)
 +
\mathbb{P}\left(\frac{1}{f(t)}\frac{(\mu-r)^2}{1-\rho^2}\int_0^t\frac{\D s}{V_s}-\frac{2\rho\lambda(\mu-r)}{1-\rho^2}\frac{t}{f(t)}
<c'_2\right)\\
& = \mathbb{P}\left(\frac{1}{f(t)}\int_0^t \gamma_1^2(s) \D s<c_1\right)
 + \mathbb{P}\left(\frac{1}{f(t)}\int_0^t\frac{\D s}{V_s}
< \frac{1-\rho^2}{(\mu-r)^2} \left[c'_2+\frac{2\rho\lambda(\mu-r)}{1-\rho^2}\frac{t}{f(t)}\right]\right),
\end{align*}
with
$c'_1=\frac{\rho^2}{1-\rho^2}c_1>0$.
As long as $c_1\in (0,a\lambda^2/b)$, the first probability tends to zero as $t$ tends to infinity by~\eqref{eq:Gamma1bPos}.
Now, when $\rho\lambda(\mu-r)<0$, then since $t/f(t)$ tends to infinity,
the second probability tends to zero (as $t$ tends to infinity) by~\eqref{eq:Gamma2bPos}
because $c'_2+\frac{2\rho\lambda(\mu-r)}{1-\rho^2}\frac{t}{f(t)}$ tends to $-\infty$
(and because the variance process is non-negative almost surely).
No condition on $c'_2$ is needed here.

When $\rho\lambda=0$, then the first line of the equation above simplifies to
$$
\mathbb{P}\left(\frac{1}{f(t)}\int_0^t\gamma_2^2(s)\D s<c_2\right)
 = \mathbb{P}\left(\frac{1}{f(t)}\frac{(\mu-r)^2}{1-\rho^2}\int_0^t\frac{\D s}{V_s}<c_2\right).
$$
From~\eqref{eq:Gamma2bPos}, it tends to zero as $t$ tends to infinity when
$0< c_2 <\frac{(\mu-r)^2 b}{(1-\rho^2)(a-\sigma)}$, and hence the proposition follows from Definition~\ref{def:MarketPrice}.
\end{proof}

We now state and prove our final result, namely a strong asymptotic arbitrage statement for the stock price
process when the speed is sublinear.
\begin{proposition}\label{Propslowerspeed}
Fix $\gamma>0$. Then, for $t$ large enough,
\begin{enumerate}
\item if $\lambda\in\RR\setminus [-\sqrt{2b\gamma/a},\sqrt{2b\gamma/a}]$,
then $S$ admits a partial $(0,1/2)$-arbitrage with speed~$f(t)$;
\item if $a>\sigma$ and $\lambda\rho(\mu-r)\leq 0$, then $S$ admits a partial $(0,1/2)$-arbitrage with speed $f(t)$,
\begin{itemize}
\item if and only if
$\lambda\in \RR\setminus\left[-\sqrt{\frac{2b\gamma(1-\rho^2)}{a\rho^2}},\sqrt{\frac{2b\gamma(1-\rho^2)}{a\rho^2}}\right]$
when $\mu = r$ and $\rho^2\leq 1/2$;
\item if and only if $\lambda\in\RR\setminus [-\sqrt{2b\gamma/a},\sqrt{2b\gamma/a}]$
when $\mu = r$ and $\rho^2 \geq 1/2$;
\item if $\mu\ne r$ and $\rho\lambda<0$;
\item if $\mu\ne r$, $\rho\lambda=0$ and $\rho^2>1-\frac{(\mu-r)^2 b}{2(a-\sigma)\gamma}$.
\end{itemize}
\end{enumerate}
\end{proposition}
\begin{proof}
Recall that we are in the framework of Proposition~\ref{prop:limProbErgodic}, so that
$c_1>0$ and $c_2>0$ are the thresholds for $\gamma_1$ and $\gamma_2$
above which $S$ has an average squared market price of risk.
In this proof, we follow steps similar to those in~\cite{Fol}.
For any $\eps_1>0$, fix $0<\gamma<\bar{\gamma}<\frac{c_1}{2}=\frac{a\lambda^2}{2b}$
and $t_0>4\bar{\gamma}/[(\bar{\gamma}-\gamma)^2\eps_1]$
such that for any $t\geq t_0$ we have
$\mathbb{P}\left(f(t)^{-1}\int_0^t\gamma_1^2(s)\D s\leq 2\bar{\gamma}\right) < \eps_1 /2$.
Define the stopping time
$\tau_1 := t\wedge\inf\left\{s\in[0,t]: \int_0^s\gamma_1^2(u) \D u \geq 2\bar{\gamma}f(t)\right\}$. Let $\tilde{t}_0>t_0$ such that for any $t\geq t_0$ we have $f(t)\geq\tilde{t}_0$. 
Then for $t\geq\tilde{t}_0$ and using the fact that $\int_0^{\tau_1}\gamma_1^2(s)\D s \leq 2\bar{\gamma} f(t)$,
Chebychev's inequality implies
$$
\mathbb{P}\left(\left|\int_0^{\tau_1}\gamma_1(s)\D W_1(s)\right| \geq (\bar{\gamma}-\gamma)f(t)\right)
 \leq \frac{2\bar{\gamma}}{(\bar{\gamma}-\gamma)^2f(t)}<\frac{\eps_1}{2}.
$$
For
$Z_{\tau_1}^1
:=\exp\left(-\int_0^{\tau_1}\gamma_1(s)\D W_1(s)-\frac{1}{2}\int_0^{\tau_1}\gamma_1^2(s) \D s\right)$,
we then obtain
\begin{align*}
\mathbb{P}\left(Z_{\tau_1}^1 \geq \E^{-\gamma f(t)}\right)
 & = \mathbb{P}\left(-\int_0^{\tau_1}\gamma_1(s)\D W_1(s)-\frac{1}{2}\gamma_1^{2}(s)\D s \geq -\gamma f(t)\right)\\
 & \leq \mathbb{P}\left(\left|\int_0^{\tau_1}\gamma_1(s)\D W_1(s)\right|\geq (\bar{\gamma}-\gamma)f(t)\right)
 + \mathbb{P}\left(\int_0^{\tau_1}\frac{\gamma_1^2(s)}{2}\D s \leq \bar{\gamma}f(t)\right)\leq \frac{\eps_1}{2}+\frac{\eps_1}{2} = \eps_1.
\end{align*}

Let $A_{\lambda,t}:=\left\{Z_{\tau_1}^1\geq \E^{-\gamma f(t)}\right\}\in\mathcal{F}_t$. We obtain $\mathbb{P}(A_{\lambda,t}) \leq \eps_1$. 
We are now in position to construct contingent claim which satisfies the arbitrage estimates of the Theorem. We can introduce the random variable
$Y_t := \eps_2\ind_{A_{\lambda,t}^c} - \eps_2\frac{\mathbb{Q}(A_{\lambda,t}^c)}{\mathbb{Q}(A_{\lambda,t})}\ind_{A_{\lambda,t}}$ 
which satisfies the properties $(i)$ and $(ii)$ in Definition~\ref{def:eps1eps2ArbNew} with $e_1=0$ and $e_2=1/2$.

Assume now that $a>\sigma$, $\lambda\rho(\mu-r)\leq 0$, then $S$ has an average squared market price of risk $\gamma_2$ above a threshold $c_2>0$.
For any $\eps_1>0$, let $0<\gamma<\gamma'<\frac{c_2}{2}$, and
$t_1>\frac{4\gamma'}{(\gamma'-\gamma)^2\eps_1}$ such that, for $t\geq t_1$,
$\mathbb{P}(f(t)^{-1}\int_0^t\gamma_2^2(s)\D s\leq 2\gamma') < \eps_1/2$.
Define the stopping time $\tau_2$ by
$\tau_2 := t\wedge \inf\{s\in[0,t]: \int_0^s\gamma_2^2(u)\D u \geq 2\gamma'f(t)\}$
and the random variable
$Z_{\tau_2}^2:=\exp(-\int_0^{\tau_2}\gamma_2(s)\D W_2(s)-\frac{1}{2}\int_0^{\tau_2}\gamma_2^2(s) \D s)$. Let $\tilde{t}_1>t_1$ such that for any $t\geq t_1$ we have $f(t)\geq\tilde{t}_1$. 
Then for $t\geq\tilde{t}_1$, we have
$\mathbb{P}(Z_{\tau_2}^2\geq\E^{-\gamma f(t)})
\leq \eps_1$.
Let $B_{\lambda,t}:=\{Z_{\tau_2}^2\geq \E^{-\gamma f(t)}\}\in\mathcal{F}_t$. We obtain $\mathbb{P}(B_{\lambda,t}) \leq \eps_1$. 
Similarly to the first case, we can introduce the random variable
$Y_t := \eps_2\ind_{B_{\lambda,t}^c} - \eps_2\frac{\mathbb{Q}(B_{\lambda,t}^c)}{\mathbb{Q}(B_{\lambda,t})}\ind_{B_{\lambda,t}}$, 
and hence $S$ satisfies a partial $(0,1/2)$-arbitrage with speed~$f(t)$.

Note that the constraint $a\lambda^2/b=c_1>2\gamma$ reads
$\lambda\in\RR\setminus [-\sqrt{2b\gamma/a},\sqrt{2b\gamma/a}]$.
The constraints on $c_2$ depend on the sign of $\lambda\rho(\mu-r)$, as explained in Remark~\ref{rem:C2}:
\begin{itemize}
\item if $\mu = r$ and $\rho^2<1/2$, then $c_1>c_2$;
then $c_2/2>\gamma$ if and only if
$\lambda\in \RR\setminus\left[-\sqrt{\frac{2b\gamma(1-\rho^2)}{a\rho^2}},\sqrt{\frac{2b\gamma(1-\rho^2)}{a\rho^2}}\right]$;
\item if $\mu = r$ and $\rho^2 > 1/2$, then $c_1<c_2$;
then $c_1/2>\gamma$ if and only if $\lambda\in\RR\setminus [-\sqrt{2b\gamma/a},\sqrt{2b\gamma/a}]$;
\item if $\mu\ne r$ and $\rho\lambda<0$, no further assumption on $\lambda$ is needed;
\item if $\mu\ne r$ and $\rho\lambda=0$, then the constraint
$0<\gamma<\frac{(\mu-r)^2 b}{2(a-\sigma)(1-\rho^2)}$ has to hold.
\end{itemize}
\end{proof}

\appendix
\section{Large deviations results}\label{App:Appendix}

\begin{proof}[Proof of Lemma~\ref{lem:triplet}]
Recall the standing assumption that $\beta$ and $\delta$ are never null simultaneously.
We first prove the lemma in the case $b\neq0$.
The moment generating function of the random variable $X_t^{\alpha,\beta,\delta}/t$
is given by (see~\cite[proposition 2]{Keb}),
\begin{align*}
\Lambda_t(u)
& =\mathbb{E}\left(\exp\left(\frac{\alpha u }{t}V_t+\frac{\beta u }{t}\int_0^t V_s\D s+\frac{\delta u }{t}\int_0^t V_s^{-1}\D s\right)\right)\\
&=
\frac{\Gamma(\kappa+\frac{1}{2}(\nu+1))}{\Gamma(\nu+1)}
\exp\left\{\frac{b}{2\sigma}(at+V_0)-\frac{AV_0}{2\sigma}\coth\left(\frac{At}{2}\right)\right\}\\
 & \times \left(\frac{AV_0}{2\sigma\sinh(At/2)}\right)^{\frac{1}{2}(\nu+1)-\kappa}
\left(\left(b-\frac{2\sigma\alpha u}{t}\right)\frac{\sinh(At/2)}{A}+\cosh(At/2)\right)^{-(\kappa+\frac{1}{2}(\nu+1))}\\
& \times _1F_1\left(\kappa+\frac{\nu+1}{2},\nu+1,\frac{A^2V_0}{2\sigma\sinh(At/2)\left((b-\frac{2\sigma\alpha u}{t})\sinh(At/2)+\cosh(At/2)\right)}\right)
\end{align*}
where $\kappa:=\frac{a}{2\sigma}$, $A:=\sqrt{b^2-\frac{4\sigma\beta u}{t}}$,
$\nu:=\frac{1}{\sigma}\sqrt{(a-\sigma)^2-\frac{4\sigma\delta u}{t}}$.
Note that we are actually extending the result of~\cite[proposition 2]{Keb} here.
Indeed, since characteristic functions can be extended in the complex plane up to the first singularity, 
the extension to positive values of $\alpha, \beta,\delta$ is trivial.
The confluent hypergeometric function is defined by
$_1F_1(u,v,z)=\sum_{n\geq 0}\frac{u^{(n)}}{v^{(n)}}\frac{z^n}{n!}$,
with $v^{(n)}$ denoting the rising factorial $v^{(n)}:=v(v+1)\ldots(v+n-1)$.
As $t$ tends to infinity,
$t^{-1}\log\left(\frac{\Gamma(\kappa+\nu/2+1/2)}{\Gamma(\nu+1)}\right)$ clearly tends to zero and
 $$\lim_{t\uparrow+\infty}\frac{1}{t}
\log\left( _1F_1\left(\kappa+\frac{\nu+1}{2},\nu+1,
\frac{A^2V_0}{2\sigma\sinh(At/2)\left[\left(b-2\sigma\alpha u\right)\sinh(At/2)+\cosh(At/2)\right]}\right)\right)=0.$$
Therefore,
\begin{align*}
\Lambda^{\beta,\delta}(u)
& := \lim_{t\uparrow+\infty}t^{-1}\log\Lambda_t(tu)\\
& = \lim_{t\uparrow+\infty}\frac{1}{t} \Bigg\{
\frac{b}{2\sigma}(at+V_0)-\frac{AV_0}{2\sigma}\frac{\E^{At/2}+\E^{-At/2}}{\E^{At/2}-\E^{-At/2}}+
\left(\frac{\nu+1}{2}-\kappa\right)\log\left(\frac{AV_0}{\sigma\left(\E^{At/2}-\E^{-At/2}\right)}\right)\\
& - \left(\kappa+\frac{\nu+1}{2}\right)\log\left(\frac{b-2\sigma\alpha u}{A}\left(\frac{\E^{At/2}-\E^{-At/2}}{2}\right)+\frac{\E^{At/2}+\E^{-At/2}}{2}\right)\Bigg\}\\
& = -\frac{\nu A}{2}-\frac{A}{2}+\frac{ba}{2\sigma}
 =\frac{ab}{2\sigma}-\frac{1}{2\sigma}\sqrt{((a-\sigma)^2-4\sigma\delta u)(b^2-4\sigma\beta u)}-\frac{1}{2}\sqrt{b^2-4\sigma\beta u},
\end{align*}
for $u\in\Dd_{\beta,\delta}$ where the interval $\Dd_{\beta,\delta}$ is given in~\eqref{eq:DomainD}.
We can then immediately compute
$$
\partial_u\Lambda^{\beta,\delta}(u)
=\frac{\sigma\beta}{\sqrt{b^2-4\sigma\beta u}}-\frac{8\sigma\delta\beta u-\beta(a-\sigma)^2-\delta b^2}{\sqrt{((a-\sigma)^2-4\sigma\delta u)(b^2-4\sigma\beta u)}}, 
\quad \text{for any }u\in\Dd^o_{\beta,\delta},
$$
and hence
\begin{equation*}
\partial_u\Lambda^{\beta,\delta}(\Dd_{\beta,\delta}^o)=
\left\{
\begin{array}{ll}
\displaystyle \mathbb{R}, & \text{if } \beta\delta < 0,\\
\displaystyle (2\sqrt{\delta\beta},+\infty), & \text{if } \beta \geq 0, \delta \geq 0,\\
\displaystyle (-\infty,-2\sqrt{\delta\beta}), & \text{if } \beta \leq 0, \delta \leq 0.
\end{array}
\right.
\end{equation*}

We also have, for any $u\in \Dd^o_{\beta,\delta}$,
$$
\partial_{uu}\Lambda^{\beta,\delta}(u)
=\frac{2\sigma^2\beta^2}{(b^2-4\sigma\beta u)^{3/2}}+\frac{2\sigma(\delta b^2 -\beta (a-\sigma)^2)^2}
{\left[((a-\sigma)^2-4\sigma\delta u)(b^2-4\sigma\beta u)\right]^{3/2}}.
$$
Therefore $\Lambda^{\beta,\delta}$ is strictly convex on $\Dd_{\beta,\delta}$, and
the G\"artner-Ellis theorem (see~\cite{DZ}) only applies on subsets of $\partial_u\Lambda^{\beta,\delta}(\Dd^o_{\beta,\delta})$.
For any $x\in\partial_u\Lambda^{\beta,\delta}(\Dd_{\beta,\delta}^o)$,
the equation $\partial_u\Lambda^{\beta,\delta}(u)=x$ has a unique solution~$u^*(x)$ and hence
$\Lambda_{\beta,\delta}^*(x):=\sup_{u\in\Dd_{\beta,\delta}}\left\{ux-\Lambda^{\beta,\delta}(u)\right\}
 = u^*(x)x-\Lambda^{\beta,\delta}(u^*(x))$.

We now move on to the case $b=0$.
From~\cite[Corollary 1]{Keb}, the moment generating function of the random variable $X_t^{\alpha,\beta,\delta}$ is given by
\begin{align*}
\Lambda_t(u) &= \mathbb{E}\left[\exp\left(\frac{\alpha u }{t}V_t+\frac{\beta u }{t}\int_0^t V_s\D s+\frac{\delta u }{t}\int_0^t V_s^{-1}\D s\right)\right]\\
&=\frac{\Gamma(\kappa+\frac{1}{2}(\nu+1))}{\Gamma(\nu+1)}
\E^{-\frac{V_0\zeta_u}{\sigma}\coth\left(\zeta_u t\right)}
\left(\frac{\zeta_uV_0}{\sigma\sinh\left(\zeta_u t\right)}\right)^{\frac{1}{2}(\nu+1)-\kappa}
 \left[\frac{-\sqrt{\sigma}\alpha u/t}{\sqrt{-\beta u/t}}
 \sinh\left(\zeta_u t\right)+\cosh\left(\zeta_u t\right)\right]^{-(\kappa+\frac{1}{2}(\nu+1))}\\
& \times _1F_1\left(\kappa+\frac{\nu+1}{2},\nu+1,\frac{V_0\zeta_u}{\sigma \sinh\left(\zeta_ut\right)
\left(-\frac{\sigma\alpha u}{\zeta_u t}\sinh(\zeta_ut)+\cosh(\zeta_ut)\right)}\right)
\end{align*}
where $\kappa:=\frac{a}{2\sigma}$, and $\nu:=\frac{1}{\sigma}\sqrt{(a-\sigma)^2-\frac{4\sigma\delta u}{t}}$,
and $\zeta_u:=\sqrt{-\frac{\sigma\beta u}{t}}$.
Similar to the case $b\neq0$, we obtain
$$
\lim_{t\uparrow+\infty}\frac{1}{t}
\log _1F_1\left(\kappa+\frac{\nu+1}{2},\nu+1,
\frac{\sqrt{-\sigma\beta}V_0}{\sigma \sinh(\sqrt{-\sigma\beta}u)
\Big(-\frac{\sqrt{\sigma}\alpha}{\sqrt{-\beta}}\sinh(\sqrt{-\sigma\beta}u)+\cosh(\sqrt{-\sigma\beta}u)\Big)}\right)=0.$$
Therefore, the limiting cumulant generating function of $X_t^{\alpha,\beta,\delta}$ reads
\begin{align*}
\Lambda^{\beta,\delta}(u)
& := \lim_{t\uparrow+\infty}t^{-1}\log\Lambda_t(tu)\\
& = \lim_{t\uparrow+\infty}\frac{1}{t}
\Bigg\{-\frac{\xi_u V_0}{\sigma}\frac{\E^{\xi_u t}+\E^{-\xi_u t}}{\E^{\xi_u t}-\E^{-\xi_u t}}
+\left(\frac{\nu+1}{2}-\kappa\right)\log\left(\frac{2\xi_u V_0}{\sigma(\E^{\xi_u t}-\E^{-\xi_u t})}\right)\\
& - \left(\kappa+\frac{\nu+1}{2}\right)\log\left[-\frac{\sqrt{\sigma}\alpha u}{\sqrt{-\beta u}}\left(\frac{\E^{\xi_u t}-\E^{-\xi_ut}}{2}\right)+\frac{\E^{\xi_u t}+\E^{-\xi_u t}}{2}\right]\Bigg\}\\
&  = -\xi_u -\frac{1}{\sigma}\xi_u \sqrt{(a-\sigma)^2-4\sigma\delta u},
\end{align*}
with $\xi_u:=\sqrt{-\sigma\beta u}$, for any $u\in\Dd_{\beta,\delta}$ where this interval now reads
\begin{equation*}
\Dd_{\beta,\delta}=
\left\{
\begin{array}{ll}
\displaystyle \left[0,\frac{(a-\sigma)^2}{4\sigma\delta}\right], & \text{if } \beta \leq 0 \text{ and } \delta > 0,\\
\displaystyle \left[\frac{(a-\sigma)^2}{4\sigma\delta},0\right], & \text{if } \beta \geq 0 \text{ and } \delta <0,\\
\displaystyle \mathbb{R}_-, & \text{if } \beta \geq 0 \text{ and } \delta \geq 0,\\
\displaystyle \mathbb{R}_+, & \text{if } \beta \leq 0 \text{ and } \delta \leq 0.
\end{array}
\right.
\end{equation*}
Then
$$
\partial_u\Lambda^{\beta,\delta}(u)
=\frac{\sigma\beta}{2\sqrt{-\sigma\beta u}}+\frac{\beta \sqrt{(a-\sigma)^2-4\sigma\delta u}}{2\sqrt{-\sigma\beta u}}
+\frac{2\delta\sqrt{-\sigma\beta u}}{\sqrt{(a-\sigma)^2-4\sigma\delta u}},
\quad\text{for any }u\in\Dd^o_{\beta,\delta},
$$
and hence
\begin{equation}\label{eq:Dd2}
\partial_u\Lambda^{\beta,\delta}(\Dd_{\beta,\delta}^o)=
\left\{
\begin{array}{ll}
\displaystyle \mathbb{R}, & \text{if } \beta\delta <0,\\
\displaystyle (2\sqrt{\delta\beta},+\infty), & \text{if } \beta \geq 0 \text{ and } \delta \geq 0,\\
\displaystyle (-\infty,-2\sqrt{\delta\beta}), & \text{if } \beta \leq 0 \text{ and }\delta \leq 0.
\end{array}
\right.
\end{equation}
We also have
$$
\partial_{uu}\Lambda^{\beta,\delta}(u)=\frac{\sigma^2\beta^2}{4(-\sigma\beta u)^{3/2}}-\frac{\beta (a-\sigma)^2}
{4u\sqrt{(a-\sigma)^2-4\sigma\delta  u}\sqrt{-\sigma\beta u}}
-\frac{\sigma\beta\delta (a-\sigma)^2}{\left((a-\sigma)^2-4\sigma\delta u\right)^{3/2}\sqrt{-\sigma\beta u}}.
$$
Clearly then, $\Lambda^{\beta,\delta}$ is convex on $\Dd_{\beta,\delta}$, and
the G\"artner-Ellis theorem only applies on subsets of $\partial_u\Lambda^{\beta,\delta}(\Dd^o_{\beta,\delta})$.
For any $x\in\partial_u\Lambda^{\beta,\delta}(\Dd^o_{\beta,\delta})$, the equation $\partial_u\Lambda^{\beta,\delta}(u)=x$ has a unique solution~$u^*(x)$ and hence
$\Lambda_{\beta,\delta}^*(x):=\sup_{u\in\Dd_{\beta,\delta}}\left\{ux-\Lambda^{\beta,\delta}(u)\right\}
 = u^*(x)x-\Lambda_{\beta,\delta}(u^*(x))$,
and the lemma follows.
\end{proof}
 \begin{lemma}\label{lem:ratefunction}
For any $x\in\partial_u\Lambda^{\beta,\delta}(\Dd_{\beta,\delta}^o)$, the equation
$\partial_u\Lambda^{\beta,\delta}(u^*(x))=x$ admits a unique solution $u^*(x)\in\Dd^o_{\beta,\delta}$.
The function~$\Lambda_{\beta,\delta}^*$ is strictly convex and satisfies
$\Lambda_{\beta,\delta}^*(x)=u^*(x)x-\Lambda^{\beta,\delta}(u^*(x))$ on $\partial_u\Lambda^{\beta,\delta}(\Dd_{\beta,\delta}^o)$
and is (positive) infinite outside.
In the case $\beta\delta\geq 0$, $\Lambda^*_{\beta,\delta}$ is strictly positive.
When $\beta\delta \leq 0$, $\Lambda_{\beta,\delta}^*$ admits a unique minimum, which is equal to zero (and is attained at the origin)
if and only if $a>\sigma$.
\end{lemma}
\begin{proof}
When $\beta\delta<0$, the image of $\Dd_{\beta,\delta}^o$ by $\partial_u\Lambda^{\beta,\delta}$ is the whole real line,
and the representation of $\Lambda_{\beta,\delta}^*$ in the lemma clearly follows.
Now, suppose there exists $\bar{x}\in\RR$ such that $\Lambda_{\beta,\delta}^*(\bar{x})=0$.
Then there exists some (possibly non-unique) $u^*(\bar{x})\in\Dd_{\beta,\delta}$
such that $u^*(\bar{x})\bar{x}=\Lambda^{\beta,\delta}(u^*(\bar{x}))$,
i.e. $\Lambda^{\beta,\delta}(u^*(\bar{x}))/u^*(\bar{x}) = \bar{x}$.
But $u^*(\bar{x})$ also satisfies $\partial_u\Lambda^{\beta,\delta}(u^*(\bar{x}))=\bar{x}$.
A straightforward analysis shows that the equality $\partial_u\Lambda^{\beta,\delta}(u) = \Lambda^{\beta,\delta}(u)/u$
is satisfied if and only if $u=0$ and $a>\sigma$.

When $\beta>0$ and $\delta>0$, for any $x\leq 2\sqrt{\beta\delta}$, the map
$u\mapsto ux-\Lambda^{\beta,\delta}(u)$ is strictly decreasing on $\Dd_{\beta,\delta}^o$, and the result follows.
By definition, the function $\Lambda_{\beta,\delta}^*$ admits a (unique) minimum $\bar x$
if and only if (i) there exists $u(\bar x)\in\Dd_{\beta,\delta}$ such that
$u(\bar x)\bar x = \Lambda^{\beta,\delta}(u(\bar x))$ and (ii) $\Lambda^{\beta,\delta}(u)>u \bar x$
for any $u\in\Dd_{\beta,\delta}\setminus\{u(\bar x)\}$.
A straightforward analysis shows that the function $u\mapsto \Lambda^{\beta,\delta}(u)/u$ on $\RR_-^*$
is strictly increasing and maps $\RR_-^*$ to $(2\sqrt{\beta\delta}, +\infty)$.
On $\RR_+^*\cap\Dd_{\beta,\delta}$, it is strictly increasing and maps this interval to $(-\infty, -2\sqrt{\beta\delta})$.
Therefore the inequality $\Lambda(u)>u x$ holds if and only if both
(a) $\Lambda^{\beta,\delta}(u)/u>x$ for $u\in\RR_+^*\cap\Dd_{\beta,\delta}$ and (b) $\Lambda^{\beta,\delta}(u)/u<x$ for $u<0$.
Case~(b) clearly only holds for $x<2\sqrt{\beta\delta}$, which is not valid.
The other cases are treated analogously.
\end{proof}




\begin{thebibliography}{99}

\bibitem{Ala}M. Ben Alaya and A. Kebaier.
Parameter estimation for the square root diffusions: ergodic and non ergodic case.
\textit{Stochastic Models}, {\tt 28}(4): 609-634, 2012.

\bibitem{Keb}M. Ben Alaya and A. Kebaier.
Asymptotic behavior of the maximum likelihood estimator for ergodic and non ergodic square-root diffusions.
\textit{Stochastic Analysis and Applications}, {\tt 31}(4): 552-573, 2013.

\bibitem{Bidima} M.L.D.~Mbele Bidima and M.~Rasonyi.
On long-term arbitrage opportunities in Markovian models of financial markets.
\textit{Annals of Operations Research}, {\tt 200}(1):131-146, 2012.

\bibitem{Del}F. Delbaen and W. Schachermayer.
The Mathematics of Arbitrage.
Springer Finance, 2006.

\bibitem{DEL}F. Delbaen and W. Schachermayer.
A general version of the fundamental theorem of asset pricing.
\textit{Math Ann}, {\tt 300}(1):463-520, 1994.

\bibitem{DZ}A.~ Dembo and O.~Zeitouni.
Large deviations techniques and applications.
Jones and Bartlet publishers, Boston, 1993.

\bibitem{Du}K. Du. and A.Neufeld.
A note on asymptotic exponential arbitrage with exponentially decaying failure probability.
\textit{Journal of Applied Probability}, {\tt 50}(3): 801-809, 2013.

\bibitem{Fol}H.  F\"{o}llmer and W. Schachermayer.
Asymptotic arbitrage and large deviations.
\textit{Mathematics and Financial Economics}, {\tt 1}(34): 213-249, 2007.

\bibitem{Fou} J.P. Fouque, G. Papanicolaou, R. Sircar and K. Solna.
Multiscale Stochastic Volatility for Equity, Interest Rate, and Credit Derivatives.
CUP, 2011.

\bibitem{Gatheral} J.~Gatheral.
The Volatility Surface: a practitioner's guide.
Wiley, 2006.

\bibitem{Heston}S.~Heston.
A closed-form solution for options with stochastic volatility with applications to bond and currency options.
\textit{The Review of Financial Studies},(6): 327-342, 1993.

\bibitem{Ber} C.C. Heyde and B. Wong.
On changes of measure in stochastic volatility models.
\textit{Journal of Applied Mathematics and Stochastic Analysis}, {\tt 2006}, Article ID 18130, 2006.

\bibitem{JYCBook}M. Jeanblanc, M. Yor and M. Chesney.
Mathematical Methods for Financial Markets
Springer Finance, 2009

\bibitem{Kabanov}Y. Kabanov and D. Kramkov.
Asymptotic arbitrage in large financial markets.
\textit{Fin. \& Stochastics}, {\tt 2}: 143-172, 1998.

\bibitem{Ire}I.~Klein and W. Schachermayer.
Asymptotic arbitrage in non-complete large financial markets.
\textit{Theory of Probability and its Applications}, {\tt 41}(4): 927-934, 1996.

\bibitem{Yur}Y.A.~Kutoyants.
Statistical inference for ergodic diffusion processes.
\textit{Springer-Verlag London}, 2004.

\bibitem{Lamberton} D. Lamberton and B. Lapeyre.
Introduction au calcul stochastique appliqu\'e à la finance, 2nd Ed.
Ellipses, Paris, 1997.

\bibitem{Lukacs}E. Lukacs.
Characteristic Functions. 
Second edition. New York: Hafner Pub. Co, 1970.

\bibitem{Mer}R.C. Merton
The theory of rational option pricing.
\textit{Bell.J.Econ.Manag.Sci}, {\tt 4}: 141-183, 1973.

\bibitem{Rev}D. Revuz and M. Yor.
Continuous martingales and Brownian motion.
Springer-Verlag, Berlin, 1999.

\bibitem{Rok}D.B. Rokhlin.
Asymptotic arbitrage and num\'{e}raire portfolio in large financial markets.
\textit{Finance and Stochastics}, {\tt 12}(2): 173-194, 2008.

\end{thebibliography}
\end{document}